\documentclass[11pt]{amsart}
\usepackage{amssymb}
\usepackage{amsmath}
\usepackage{amscd}
\begin{document}
\setlength{\oddsidemargin}{0cm} \setlength{\evensidemargin}{0cm}

\theoremstyle{plain}

\newtheorem{theorem}{Theorem}[section]
\newtheorem{prop}[theorem]{Proposition}
\newtheorem{defn}{Definition}[section]
\newtheorem{lemma}[theorem]{Lemma}
\newtheorem{coro}[theorem]{Corollary}
\newtheorem{claim}{Claim}[section]
\newtheorem{remark}[theorem]{Remark}
\newtheorem{conjecture}{Conjecture}
\newtheorem{exam}[theorem]{Example}
\newtheorem{assumption}[theorem]{Assumption}
\newtheorem{condition}[theorem]{Condition}

\title{Isotropic ideals of metric $n$-Lie algebras}

\author{Ruipu Bai}
\address{College of Mathematics and Computer Science, Hebei University, Baoding
071002, P.R. China}\email{bairp1@yahoo.com.cn}

\author{Wanqing Wu}
\address{College of Mathematics and Computer Science, Hebei University, Baoding
071002, P.R. China} \email{wuwanqing8888@126.com}

\author{Zhiqi Chen}
\address{School of Mathematical Sciences and LPMC, Nankai University,
Tianjin 300071, P.R. China} \email{chenzhiqi@nankai.edu.cn}

\def\shorttitle{Isotropic ideals of metric $n$-Lie algebras}

\subjclass[2000]{17B05, 17D99}


\begin{abstract}
In this paper, we give a systematic study on isotropic ideals of
metric $n$-Lie algebras. As an application, we show that the center
of a non-abelian $(n+k)$-dimensional metric $n$-Lie algebra ($2\leq
k\leq n+1$), whose center is isotropic, is of dimension $k-1$.
Furthermore, we classify $(n+k)$-dimensional metric $n$-Lie algebras
for $2\leq k\leq n+1$.
\end{abstract}

\maketitle


\baselineskip=18pt

\section{Introduction}

The notion of $n$-Lie algebras was introduced by Filippov \cite{1}
in 1985. It is strongly connected with many other fields, such as
dynamics, geometry and string theory. Indeed, motivated by some
problems of quark dynamics, Nambu \cite{2} introduces a ternary
generalization of Hamiltonian dynamics by means of the ternary
Poisson bracket
$$ [f_1, f_2, f_3]= \det\Big(\frac{\partial f_i}{\partial x_j}\Big).$$
Recently, a class of $n$-Lie algebras attract more and more
attention. They are called metric $n$-Lie algebras, which are
$n$-Lie algebras possessed invariant non-degenerate symmetric
bilinear forms.

 Metric $n$-Lie algebras have arisen
 in the work of Figueroa-O'Farrill and Papadopoulos
\cite{3} in the classification of maximally supersymmetric type IIB
supergravity backgrounds \cite{4}, and in the work of Bagger and
Lambert \cite{5,6} and Gustavsson \cite{7} on a superconformal field
theory for multiple M2-branes. But very few metric $n$-Lie algebras
admit positive definite metrics (see \cite{8,9}). In fact,
Ho-Hou-Matsuo \cite{10} confirmed that there are no $n$-Lie algebras
with positive definite metrics for $n = 5, 6, 7, 8$. Some messages
in physics motivate us to study $n$-Lie algebras with invariant
non-degenerate symmetric bilinear forms, such as the correspond
between some dynamical systems involving zero-norm generators and
gauge symmetries and negative-norm generators corresponding to
ghosts (see \cite{11,12,13}).

As we know, there is some progress on the structures of metric
$n$-Lie algebras, such as the metric dimension of metric $n$-Lie
algebras \cite{14}, Lorentzian metric $n$-Lie algebras \cite{15},
the double extension of $n$-Lie algebras \cite{16} and the
classification of index-$2$ metric $3$-Lie algebras \cite{12}. But
there are many problems unsolved. In particular, the classification
of metric $n$-Lie algebras is also open. In this paper, we mainly
concern isotropic ideals of metric $n$-Lie algebras. Furthermore
based on the results about isotropic ideals, we classify
$(n+k)$-dimensional metric $n$-Lie algebras for $2\leq k\leq n+1$.

The paper is organized as follows. In Section 2, we list some
fundamental notions. In Section 3, we study isotropic ideals of
metric $n$-Lie algebras. In Section 4, we  classify
$(n+k)$-dimensional metric $n$-Lie algebras for $2\leq k\leq n+1$.

Throughout this paper, all $n$-Lie algebras are of finite dimension
and over the complex field $\mathbb C$. Obvious proof are omitted.

\section{Fundamental notions}

\subsection{Basic notions on $n$-Lie algebras.}
An $n$-Lie algebra is a vector space $\mathfrak g$ over a field
$\mathbb C$ equipped with an $n$-multilinear operation $[x_1,
\cdots, x_n]$ satisfying
$$ [x_1, \cdots, x_n] = sign(\sigma)[x_{\sigma (1)}, \cdots, x_{\sigma(n)}], \eqno(2.1) $$  $$
  [[x_1, \cdots, x_n], y_2, \cdots, y_n]=\sum_{i=1}^n[x_1, \cdots, [ x_i, y_2, \cdots, y_n], \cdots, x_n]\eqno(2.2) $$
for any $x_1, \cdots, x_n, y_2, \cdots, y_n\in {\mathfrak g}$ and
$\sigma\in S_n$, the permutation group on $n$ letters. The
subalgebra $[{\mathfrak g}, \cdots, {\mathfrak g}]$ generated by all
vectors $[x_1, \cdots, x_n]$ for any $x_1, \cdots, x_n\in {\mathfrak
g}$ is called the derived algebra of $ {\mathfrak g}$, denoted by $
{\mathfrak g}^1$. If ${\mathfrak g}^1 = {\mathfrak g}$, then
${\mathfrak g}$ is called a perfect $n$-Lie algebra.

An ideal $I$ of an $n$-Lie algebra ${\mathfrak g}$ is called a
solvable ideal, if $I^{(r)} = 0 $ for some $r \geq 0$, where
$I^{(0)} = I$ and $I^{(s+1)} = [I^{(s)}, \cdots, I^{(s)}]$ for $s >
0$ by induction. If $[I, I, {\mathfrak g}, \cdots, {\mathfrak
g}]=0,$ then $I$ is called an abelian ideal. The maximal solvable
ideal of ${\mathfrak g}$ is called the solvable radical, denoted by
${\mathfrak r}$.

An $n$-Lie algebra ${\mathfrak g}$ is said to be simple if it has no
proper ideals and $\dim {\mathfrak g}^1 > 0$. If ${\mathfrak g}$ is
the direct sum of simple ideals, then ${\mathfrak g}$ is strong
semisimple (see \cite{18}). By \cite{19}, there exists only one
finite dimensional simple $n$-Lie algebra ${\mathfrak g}$ over
$\mathbb C$, which is the perfect $(n+1)$-dimensional $n$-Lie
algebra; and every finite dimensional $n$-Lie algebra ${\mathfrak
g}$ has a Levi-decomposition
 $${\mathfrak g}={\mathfrak s}\oplus {\mathfrak r},$$ where ${\mathfrak s}$ is a strong semisimple
subalgebra and ${\mathfrak r}$ is the radical of ${\mathfrak g}$.

For a given subspace $W$ of an  $n$-Lie algebra ${\mathfrak g}$, the
subalgebra
$$
 C_{\mathfrak g}(W)=\{x\in {\mathfrak g} \mid [x,
W, {\mathfrak g}, \cdots, {\mathfrak g}]=0\}
$$
 is called the centralizer of $W$ in ${\mathfrak g}$. The
 centralizer of ${\mathfrak g}$, i.e., the center of ${\mathfrak g}$, is denoted by $C({\mathfrak g})$. If $I$ is an ideal of
 ${\mathfrak g}$, so is $C_{{\mathfrak g}}(I)$.

An $n$-Lie module of an $n$-Lie algebra ${\mathfrak g}$ is a vector
space $M$ with a multilinear skew-symmetric mapping~$\rho:{\mathfrak
g}^{\wedge {(n-1)}} \longrightarrow End (M)$ ~ satisfying
$$
[\rho(a),
\rho(b)]=\rho(a)\rho(b)-\rho(b)\rho(a)=\sum\limits_{i=1}^{n-1}\rho(b_1,
\cdots, [a_1, \cdots, a_{n-1}, b_i], \cdots, b_{n-1}),
$$
$$
\rho([a_{1}, \cdots, a_{n}], b_{2}, \cdots, b_{n-1})=
$$
$$
\sum\limits_{i=1}^{n}(-1)^{n-i}\rho(a_{1}, \cdots, a_{i-1}, a_{i+1},
\cdots, a_{n})\rho(a_i, b_2, \cdots, b_{n-1}),
$$
where $ a=(a_{1}, \cdots, a_{n-1})$,  $ b=(b_1, \dots, b_{n-1})\in
{\mathfrak g}^{\wedge (n-1)}$ (for more details see
\cite{20,21,22}). The regular representation on ${\mathfrak g}$ is
defined by
 $$
 a\mapsto ad(a), \mbox{ for }  a=(a_1, \cdots a_{n-1})\in  {\mathfrak g}^{ \wedge (n-1)}.
 $$

\subsection{Basic facts on metric $n$-Lie algebras.}
A metric $n$-Lie algebra $({\mathfrak g}, B)$ is an $n$-Lie algebra
${\mathfrak g}$ with non-degenerate symmetric bilinear form $B$,
satisfying
$$B([x_{1}, \cdots, x_{n-1}, y_{1}], y_{2})=-B([x_{1}, \cdots, x_{n-1}, y_{2}], y_{1}),
~\forall x_{1},\cdots, x_{n-1}, y_1,  y_{2}\in  {\mathfrak
g}.\eqno(2.3)$$ Such a bilinear form $B$ is called a metric on
${\mathfrak g}$.

Let $W$ be a subspace of a metric $n$-Lie algebra $({\mathfrak g},
B)$. The orthogonal complement of $W$ is
$$
    W^{\bot}=\{x \in  {\mathfrak g} \mid B(w, x)=0,~\mbox{for any} ~ w\in W\}.
$$
 Then ${\mathfrak g}^1=[{\mathfrak g},
\cdots, {\mathfrak g}]=C({\mathfrak g})^{\bot}$. The subspace $W$ is
called non-degenerate if $B|_{W \times W}$ is non-degenerate, that
is $W\cap W^{\bot}=0$. If $W$ is an ideal, then $W^{\bot}$ is also
an ideal, and $W$ is a minimal ideal if and only if $W^{\bot}$ is
maximal. A subspace $W$ is isotropic (coisotropic) if $W \subseteq
W^{\bot}$ ($W^{\bot} \subseteq W$).

A metric $n$-Lie algebra ${\mathfrak g}$ is  decomposable if there
exist non-degenerate ideals $I$ and $J$, such that ${\mathfrak
g}=I\oplus J$. Otherwise, ${\mathfrak g}$ is indecomposable. An
equivalent criterion for decomposability is the existence of a
proper non-degenerate ideal: if $I$ is a non-degenerate ideal,
${\mathfrak g}=I\oplus I^{\bot}$ is an orthogonal direct sum of
ideals.

\begin{defn} Let $({\mathfrak g},  {B})$ be a metric $n$-Lie algebra
and ${I}$ an ideal of ${\mathfrak g}$. If $ {I}$ is isotropic and is
not contained in (does not contain) any isotropic ideal, then ${I}$
is called an isomaximal ideal (an isominimal ideal) of ${\mathfrak
g}$.
\end{defn}

\begin{lemma} Let $({\mathfrak g}, B)$ be a metric $n$-Lie
algebra and $V$, $W$ subspaces of ${\mathfrak g}$. Then $\dim G=\dim
W+\dim W^{\bot}$ and
$$ V \subseteq W
\Leftrightarrow V^{\bot} \supseteq W^{\bot},\quad V^{\bot} \cap
W^{\bot} \subseteq (V + W)^{\bot},
$$
$$
(V \cap W)^{\bot} \supseteq V^{\bot} +
W^{\bot},\quad(V^{\bot})^{\bot}=V.
$$
\end{lemma}

\begin{lemma} Let $({\mathfrak g}, B)$ be a metric $n$-Lie algebra and $J_{1}$,
$J_{2}$ isotropic subspaces of $G$. Then $J_{1}+J_{2}$ is isotropic
if and only if $B(J_1, J_{2})=0$.
\end{lemma}

\begin{lemma}[\cite{14}] Let $({\mathfrak g}, B)$ be a metric $n$-Lie algebra,
${\mathfrak g}={\mathfrak s}\oplus {\mathfrak r}$ the
Levi-decomposition. Then ${\mathfrak r}^{\bot}$ is isotropic if and
only if ${\mathfrak g}$ has no strong semisimple ideals. Moreover,
we have
$$
    C_{{\mathfrak g}}({\mathfrak r})=C({\mathfrak g}) \oplus  {\mathfrak r}^{\bot}
    \quad\mbox{and}\quad[{\mathfrak s}, \cdots, {\mathfrak s}, {\mathfrak r}^{\bot}]={\mathfrak r}^{\bot}.
$$
\end{lemma}

\begin{lemma}[\cite{14}] Let $W$ be a subspace of a metric $n$-Lie
algebra $({\mathfrak g}, B)$ and $I$ an ideal. Then
\begin{enumerate}
\item $C_{{\mathfrak g}}(I)=[I, {\mathfrak g}, \cdots,  {\mathfrak g}]^{\bot}.$
\item $W$ is an ideal of ${\mathfrak g}$ if and only if $W^{\bot}$ is contained in $C_{{\mathfrak g}}(W)$.
\item One-dimensional non-degenerate ideals are contained in $C({\mathfrak g})$.
\end{enumerate}
\end{lemma}

\begin{prop} Let $({\mathfrak g}, B)$ be a metric $n$-Lie algebra and
$I$ a subalgebra (an ideal) of ${\mathfrak g}$. Then $I$ is
isotropic if and only if there exist linear endomorphisms $D$ and
$T$ of $I$ such that $B(D(x), y)=B(x, T(y))$ for any $x, y\in I$,
where $D$ is surjective and $T$ is nilpotent.
\end{prop}

\begin{proof} Assume that $I$ is isotropic. It is enough to let $D(x)=x$ and $T(x)=0$ for
any $x\in I$. Conversely, since $T$ is nilpotent, there exists an
integer $m$ such that $T^m=0$. Then any $x, y\in {\mathfrak g},$
$B(D^m(x), y)=B(x, T^m(y))=0$. Therefore, $B(I, I)=0$ since
$D^m(I)=I$.
\end{proof}

\section{Isotropic ideals of metric $n$-Lie algebras}
Firstly, we list some lemmas.
\begin{lemma} Let $({\mathfrak g}, B)$ be a metric $n$-Lie algebra,
${\mathfrak g}={\mathfrak s}\oplus {\mathfrak r}$ the
Levi-decomposition and $I$ an ideal.
\begin{enumerate}
\item If $I$ is isotropic, then $I$ is an abelian ideal of ${\mathfrak g}$.
\item If $I$ is a non-degenerate abelian ideal, then $I\subseteq C({\mathfrak g}).$
\item The center $C({\mathfrak g})$ is isotropic if and only if $C({\mathfrak g})\subseteq
{\mathfrak g}^1.$
\item ${\mathfrak g}$ is perfect if and only if $C({\mathfrak g})=0$.
\item ${\mathfrak g}$ has no strong semi-simple ideals if and only if ${\mathfrak r}^{\bot}
\subseteq C({\mathfrak r})$. Here $C({\mathfrak r})$ denotes the
center of ${\mathfrak r}$.
\end{enumerate}
\end{lemma}

\begin{lemma} Let $({\mathfrak g}, B)$ be a metric n-Lie algebra, $I$
an isotropic ideal satisfying $I^{\bot}\neq I$. Then $(I^{\bot}/I,
B)$ is a metric $n$-Lie algebra.
\end{lemma}

\begin{lemma} Let $({\mathfrak g}, B)$ be a metric $n$-Lie algebra, $I$ a minimal ideal. Then $I$ is isotropic
if and only if $I$ is degenerate.
\end{lemma}
\begin{proof}
The minimality of the ideal $I$ implies $I\cap I^{\bot}=0$ or $I\cap
I^{\bot}=I$. If $I$ is degenerate, i.e., $I\cap I^{\bot}\neq 0$,
then $I\cap I^{\bot}=I \subseteq I^{\bot}.$  The converse is
trivial.
\end{proof}

\begin{theorem} Let $({\mathfrak g}, B)$ be a metric $n$-Lie algebra with
non-trivial solvable radical ${\mathfrak r}$ and $I$ an isomaximal
ideal. Then we have
$${\mathfrak r} \cap {\mathfrak r}^{\bot} \subseteq I
\subseteq {\mathfrak r}.$$ Furthermore assume that $({\mathfrak g},
B)$ is indecomposable. Then
\begin{enumerate}
\item $I^{\bot} \subseteq {\mathfrak r}$.
\item $I^{\bot} /I$ is an abelian subalgebra of ${\mathfrak g}/I$.
\item $I^{\bot}/I=[I^{\bot}/I, I^{\bot}/I, {\mathfrak g}/I,
\cdots, {\mathfrak g}/I] \oplus C_{{\mathfrak g}/I}((I^{\bot}/I)\cap
(I^{\bot}/I)).$
\end{enumerate}
\end{theorem}

\begin{proof}
Let ${\mathfrak g}={\mathfrak s} \oplus {\mathfrak r}$ be the Levi
decomposition. Since the ideal $I$ is isotropic, by Lemma 3.1, $I$
is abelian. Then $I\subseteq {\mathfrak r}$ and $B(I, {\mathfrak
r}\cap {\mathfrak r}^{\bot})=0$. Thanks to Lemma 2.2, $I+({\mathfrak
r} \cap {\mathfrak r}^{\bot})$ is an isotropic ideal. This implies
$I+({\mathfrak r} \cap {\mathfrak r}^{\bot})=I$, i.e., ${\mathfrak
r} \cap {\mathfrak r}^{\bot} \subseteq I$.

If $({\mathfrak g}, B)$ is indecomposable, then ${\mathfrak
r}^{\bot}\subseteq {\mathfrak r}$ by Lemma 3.1. Together with Lemma
2.1,
$$I^{\bot}\subseteq ({\mathfrak r} \cap {\mathfrak
r}^{\bot})^{\bot}={\mathfrak r} + {\mathfrak r}^{\bot} \subseteq
{\mathfrak r}.$$

Furthermore, $I^{\bot}/I$ is a solvable subalgebra of ${\mathfrak
g}/I$. Set
 $$Z=\{ x\in I^{\bot} | ~[x, u_1, \cdots, u_{n-1}]\in I, ~\mbox{for any}~ u_1, \cdots, u_{n-1}\in
I^{\bot}\}.$$
 Then $ {Z}$ is an ideal of
${\mathfrak g}$,  ${I}\subseteq {Z}$, and
 $~ {Z}/ {I}= {C}
 ( {I}^{\bot}/ {I}).$
By Lemma 3.2,  $$~Z/I=[
 I^{\bot}/I, \cdots,  I^{\bot}/I]^{\bot}.\eqno(3.1)$$

Since $B([I^{\bot}, \cdots, I^{\bot}], Z)=B(I^{\bot}, [Z, I^{\bot},
\cdots, I^{\bot}])=0$, we have $[I^{\bot}, \cdots, I^{\bot}]\cap Z$
is an isotropic ideal of ${\mathfrak g}$. Thanks to $B([I^{\bot},
\cdots, I^{\bot}]\cap Z, I)=0$ and Lemma 2.2, $ ([I^{\bot}, \cdots,
I^{\bot}]\cap Z)+I$ is an isotropic ideal of ${\mathfrak g}$.
Therefore $[{I}^{\bot}, \cdots, {I}^{\bot}]\cap {Z}\subseteq {I}.$
It follows that $[{I}^{\bot}/{I}, \cdots,
 {I}^{\bot}/{I}]\cap ({Z}/
{I})=0.$ By the identity $(3.1)$, we obtain
$$
 {I}^{\bot}/ {I}=[ {I}^{\bot}/ {I}, \cdots,
 {I}^{\bot}/ {I}] \oplus ( {Z}/
{I})=( {I}^{\bot}/ {I})^{1} \oplus ( {Z}/ {I}).\eqno(3.2)$$
 Therefore,
 $$ ({I}^{\bot}/ {I})^{1}=[ {I}^{\bot}/ {I}, \cdots,
 {I}^{\bot}/ {I}]=[( {I}^{\bot}/ {I})^{1}, \cdots, (
{I}^{\bot}/ {I})^{1}] =({I}^{\bot}/{I})^{(2)}. $$
 Since ${I}^{\bot}/ {I}$ is solvable, we have $({I}^{\bot}/ {I})^1=0$, i.e., ${I}^{\bot}/ {I}$ is abelian.

For the last assertion, let
$$ {Z}_1=\{ x\in {I^{\bot}} | ~[x, y, u_1,
\cdots, u_{n-2}]\in  {I}, ~\mbox{for any}~ y\in {I}^{\bot}, ~u_1,
\cdots, u_{n-2}\in  {\mathfrak g} \}.$$ Then ${I}\subseteq {Z}_1$,
${Z}_1$ is an ideal of ${\mathfrak g}$, and $ {Z}_1/ {I}= ({C}_{
{\mathfrak g}/{I}}( {I}^{\bot}/ {I}))\cap ( {I}^{\bot}/ {I}).$
Similarly, we have
$${Z}_1/ {I}=[
  {I}^{\bot}/ {I},  {I}^{\bot}/ {I},  {\mathfrak g}/ {I}, \cdots,  {\mathfrak g}/ {I}]^{\bot}$$
 and $ ([ {I}^{\bot},  {I}^{\bot},
 {\mathfrak g}, \cdots,  {\mathfrak g}]\cap {Z}_1)+ {I}$
is an isotropic ideal of $ {\mathfrak g}$. Thus $[ {I}^{\bot},
{I}^{\bot},
 {\mathfrak g}, \cdots,  {\mathfrak g}]\cap {Z}_1\subseteq
 {I}.$ It follows that $[ {I}^{\bot}/ {I},
 {I}^{\bot}/{I},  {\mathfrak g}/ {I}, \cdots,
 {\mathfrak g}/ {I}]\cap ( {Z}_1/ {I})=0.$
Therefore, the theorem follows.
\end{proof}

\begin{coro} If $({\mathfrak g}, B)$ is an indecomposable metric $n$-Lie algebra
with the non-zero isomaximal center, then ${\mathfrak g}$ is
solvable.
\end{coro}
\begin{proof}
  By Theorem 3.4, the derived algebra ${\mathfrak g}^1\subseteq {\mathfrak r}$.
  It follows that ${\mathfrak s}=0.$
\end{proof}

\begin{theorem} Let $({\mathfrak g}, B)$ be an indecomposable metric $n$-Lie
algebra without strong semi-simple ideals and ${\mathfrak
g}={\mathfrak s} \oplus {\mathfrak r}$ the Levi decomposition. Then
one and only one of the following cases holds up to isomorphisms.
\begin{enumerate}
 \item The radical ${\mathfrak r}$ is an isomaximal ideal, and ${\mathfrak r}$ is isomorphic
to ${\mathfrak s}$ as ${\mathfrak s}$-module in the regular
representation.
 \item There exists a nonzero vector subspace ${V}$ of ${\mathfrak r}$
 satisfying
 $$[{V}, {\mathfrak s}, \cdots,  {\mathfrak s}] \subseteq  {V}, ~
~ {\mathfrak r}={\mathfrak r}^{\bot} \oplus  {V},~ ~
  {B}( {\mathfrak s},  {V})=0,$$
 $ {B} |_{ {V} \times  {V}}$ and $B_{({\mathfrak s}\oplus {\mathfrak r}^{\bot})\times
({\mathfrak s}\oplus {\mathfrak r}^{\bot})}$ are non-degenerate.
\end{enumerate}
\end{theorem}

\begin{proof} If ${\mathfrak r}$ is isotropic, then ${\mathfrak r}\subseteq {\mathfrak r}^{\bot}$. By Lemma~2.3, ${\mathfrak r}={\mathfrak r}^{\bot}$. For every
isomaximal ideal $I$, thanks to Theorem 3.4, ${\mathfrak r}\cap
{\mathfrak r}^{\bot}\subseteq I\subseteq {\mathfrak r}$. It follows
that ${\mathfrak r}$ is an isomaximal ideal. By Theorem 3.6 in [6],
${\mathfrak r}$ is isomorphic to ${\mathfrak s}$ as a ${\mathfrak
s}$-module.

Assume that ${\mathfrak r}$ is non-isotropic. By Lemma~2.3,
${\mathfrak r}^{\bot}$ is properly contained in ${\mathfrak r}$. Let
${L}= {\mathfrak s} \oplus {\mathfrak r}^{\bot}$.

Firstly, we have ${B} |_{ {L} \times {L}}$ is non-degenerate. In
fact, let $x=s+z\in  {L}$ satisfying $ {B}(x, {L})=0$, where $s\in
{\mathfrak s}, z\in {\mathfrak r}^{\bot}$. Clearly, ${B}(s+z,
{\mathfrak r}^{\bot})= {B}(s, {\mathfrak r}^{\bot})=0$. Namely,
$s\in {\mathfrak r}\cap {\mathfrak s}$. That is, $s=0$. 
Similarly, $z=0$. Then ${B} |_{ {L} \times {L}}$ is non-degenerate.

Furthermore we have $ {B}|_{ {L}^{\bot}\times
 {L}^{\bot}}$ is non-degenerate and
$$ {\mathfrak g}= {L} \oplus  {L}^{\bot}, 
  ~  {B}( {\mathfrak s}, ~
{L}^{\bot})=0, ~  {B}( {\mathfrak r}^{\bot}, ~ {L}^{\bot})=0, ~
{\mathfrak r}= {\mathfrak r}^{\bot} \oplus  {L}^{\bot}.$$ It follows
that $B([{\mathfrak s}, \cdots, {\mathfrak s}, L^{\bot}], {\mathfrak
s})=B(L^{\bot}, [{\mathfrak s}, \cdots, {\mathfrak s}])=0$ and
$B([{\mathfrak s}, \cdots, {\mathfrak s}, L^{\bot}],
 {\mathfrak r}^{\bot})=0$. That is,
$[{\mathfrak s}, \cdots,  {\mathfrak s},
 {L}^{\bot}]\subseteq  {L}^{\bot}.$
\end{proof}

It is well known that the metric dimension of strong semisimple
$n$-Lie algebra ${\mathfrak s}$ is $m$ if ${\mathfrak s}={\mathfrak
s}_1\oplus {\mathfrak s}_2\oplus \cdots\oplus {\mathfrak s}_m$,
where ${\mathfrak s}_i$ are simple ideals. The following theorem is
to determine the metric dimension of metric $n$-Lie algebras with
isomaximal radicals.

\begin{theorem}  Let $( {\mathfrak g},  {B})$ be a perfect metric $n$-Lie
algebra without semisimple ideals and $ {\mathfrak g}= {\mathfrak s}
\oplus {\mathfrak r}$ the Levi decomposition. Then
 $\dim {\mathfrak g} \geq 2(n+1).$ Furthermore assume that the radical ${\mathfrak r}$ is isotropic. Then
\begin{enumerate}
\item ${\mathfrak r}$ is the unique isomaximal ideal.
\item $\dim {\mathfrak g}=2m(n+1)$ and ${\mathfrak g}={\mathfrak g}_1\oplus \cdots\oplus {\mathfrak g}_m,$
where ${\mathfrak g}_i={\mathfrak s}_i\oplus {\mathfrak r}_i$ (as
the direct sum of subalgebras) are ideals of ${\mathfrak g}$,
$${\mathfrak s}={\mathfrak s}_1\oplus \cdots \oplus {\mathfrak s}_m, \quad {\mathfrak r}={\mathfrak r}_1\oplus \cdots \oplus {\mathfrak r}_m,$$
 and ${\mathfrak r}_i$ is isomorphic to ${\mathfrak s}_i$ as
${\mathfrak s}_i$-module in the regular representation, $1\leq i\leq
m$.
\item $B({\mathfrak g}_i, {\mathfrak g}_j)=0$ for $1\leq i\neq j\leq
m$.
\item The metric dimension of ${\mathfrak g}$ is $2m$. Moreover there
exists a basis of ${\mathfrak g}$ such that the matrix form of every
metric $B$ associated with the basis is
$$B=diag(B_1, \cdots, B_m), \quad B_i=\left(
      \begin{array}{cccc}
        \lambda_i I_0 & \mu_i I_0  \\
     \mu_i I_0 & 0  \\
        \end{array}
    \right),  \eqno(3.3)$$  where $\lambda_i,\mu_i\in {\mathbb C}$, $\lambda_i\mu_i\neq 0$, $1\leq i\leq m$, and $I_0$ is the $(n+1)\times (n+1)$-unit
    matrix.
\end{enumerate}
\end{theorem}

\begin{proof} Since ${\mathfrak g}$ is perfect and has no strong semisimple ideals, we
have $${\mathfrak g}^1=[{\mathfrak g}, \cdots, {\mathfrak
g}]={\mathfrak g},\quad {\mathfrak r}^{\bot}\subseteq {\mathfrak
r}\neq 0, \quad {\mathfrak s}\neq 0.$$ Thanks to Lemma 2.3,
$[{\mathfrak s}, \cdots, {\mathfrak s}, {\mathfrak
r}^{\bot}]={\mathfrak r}^{\bot}\neq 0$. By the representation
properties given by \cite{20,22}, we have $\dim {\mathfrak
r}^{\bot}> n.$ Therefore, $\dim {\mathfrak g} \geq 2(n+1)$.

Assume that ${\mathfrak r}$ is isotropic. Then the result (1)
follows from Theorem 3.4.

Now let ${\mathfrak s}={\mathfrak s}_1\oplus \cdots\oplus {\mathfrak
s}_m$ be the decomposition of the strong semisimple subalgebra of
${\mathfrak g}$. By Theorem 3.6 in [6], ${\mathfrak r}$ is
isomorphic to ${\mathfrak s}$ as a ${\mathfrak s}$-module in the
regular representation. Then
$${\mathfrak r}={\mathfrak r}_1\oplus {\mathfrak r}_2\oplus \cdots\oplus {\mathfrak r}_m,$$
where
$$[{\mathfrak s}_i, \cdots, {\mathfrak s}_i, {\mathfrak r}_j]=\delta_{ij}{\mathfrak r}_j, ~1\leq i, j\leq m, \eqno(3.4)$$ that
is, ${\mathfrak r}_i$ is isomorphic to ${\mathfrak s}_i$ as an
${\mathfrak s}_i$-module in the regular representation of
${\mathfrak s}_i$.

Set ${\mathfrak g}_i={\mathfrak s}_i\oplus {\mathfrak r}_i,$ ~$1\leq
i\leq  m.$ For every $i\neq k,$
$$B([{\mathfrak s}_i, {\mathfrak s}_k, {\mathfrak s}, \cdots, {\mathfrak s}, {\mathfrak r}], {\mathfrak r})=
B([{\mathfrak s}_i, {\mathfrak s}_k, {\mathfrak s}, \cdots,
{\mathfrak s}, {\mathfrak r}], {\mathfrak r}^{\bot})=0,
$$
$$B([{\mathfrak s}_i, {\mathfrak s}_k, {\mathfrak s}, \cdots, {\mathfrak s}, {\mathfrak r}], {\mathfrak s})=
B({\mathfrak r}, [{\mathfrak s}_i, {\mathfrak s}_k, {\mathfrak s},
\cdots, {\mathfrak s}])=0. $$ It follows that $[{\mathfrak s}_i,
{\mathfrak s}_k, {\mathfrak s}, \cdots, {\mathfrak s}, {\mathfrak
r}]=0.$ Then
$$[{\mathfrak s}_i, {\mathfrak s}, {\mathfrak s}, \cdots,
{\mathfrak s}, {\mathfrak r}]=[{\mathfrak s}_i, \cdots, {\mathfrak
s}_i, {\mathfrak r}_i]={\mathfrak r}_i.$$ Therefore ${\mathfrak
g}_i$ is an ideal of ${\mathfrak g}$. This proves the result (2).

For every $1\leq i\neq j\leq m,$
$$B({\mathfrak s}_i, {\mathfrak s}_j)=B([{\mathfrak s}_i, \cdots, {\mathfrak s}_i], {\mathfrak s}_j)=B({\mathfrak s}_i, [{\mathfrak s}_i,
\cdots, {\mathfrak s}_i, {\mathfrak s}_j])=0, $$
$$B({\mathfrak r}_i, {\mathfrak r})=0,\quad B({\mathfrak r}_i, {\mathfrak s}_j)=B({\mathfrak r}_i, [{\mathfrak s}_j, \cdots, {\mathfrak s}_j])
=B([{\mathfrak r}_i, {\mathfrak s}_j, \cdots,  {\mathfrak s}_j],
{\mathfrak s}_j)=0.$$ This implies $B({\mathfrak g}_i, {\mathfrak
g}_j)=0$ for $i\neq j$.

Assume that ${\mathfrak g}_0={\mathfrak s}_0\oplus {\mathfrak r}_0$
is an indecomposable perfect $n$-Lie algebra, where ${\mathfrak
s}_0$ is the simple subalgebra and the radical ${\mathfrak r}_0$ is
isomorphic to ${\mathfrak s}_0$ as the ${\mathfrak s}_0$-module in
the regular representation. Thanks to the result (2), $\dim
{\mathfrak r}_0=n+1$. Let $x_1, \cdots x_{n+1}$ be a basis of
${\mathfrak s}_0$ satisfying
$$
[x_1, \cdots, \hat{x}_i, \cdots, x_{n+1}]=(-1)^ix_i, ~ 1\leq i\leq
n+1.
$$
Let $y_1, \cdots, y_{n+1}$ be a basis of ${\mathfrak r}_0$
satisfying
$$
[x_1, \cdots, \hat{x}_i, \cdots, \hat{x}_j , \cdots, x_{n+1},
y_j]=(-1)^{n-j+i}y_i, ~ 1\leq i <j\leq n+1;
$$
$$
[x_1, \cdots, \hat{x}_i, \cdots, \hat{x}_j , \cdots, x_{n+1},
y_k]=0, ~ 1\leq i \neq j\neq k\neq i \leq n+1.
$$
Assume that $B$ is a symmetric invariant bilinear form on
${\mathfrak g}_0$. Then
$$
B(x_i, x_j)=B((-1)^i[x_1, \cdots, \hat{x}_i, \cdots, x_{n+1}], x_j)
$$
$$=B(x_k, (-1)^{n-k-1}[x_1, \cdots, \hat{x}_i, \cdots,
\hat{x_k}, \cdots, x_{n+1}, x_j])=0, k\neq i\neq j\neq k.$$
$$B(x_i, y_j)=B([x_1, \cdots, \hat{x}_i, \cdots, x_{n+1}], y_j)
$$
$$=B(x_k, (-1)^{n-k+1}[x_1, \cdots, \hat{x}_i, \cdots,
\hat{x_k}, \cdots, x_{n+1}, y_j])=0, k\neq i\neq j\neq k.$$
$$
B(x_i, x_i)=B((-1)^i[x_1, \cdots, \hat{x}_i, \cdots, x_{n+1}], x_i)
$$
$$=B(x_j, (-1)^{n-j+i+1}[x_1, \cdots, \hat{x}_i, \cdots,
\hat{x}_j, \cdots, x_{n+1}, x_i])=B(x_j, x_j), i\neq j.$$
$$
B(x_i, y_i)=B((-1)^{i}[x_1, \cdots, \hat{x}_i, \cdots, x_{n+1}],
y_i)
$$
$$=B(x_j, (-1)^{n-j+i+1}[x_1, \cdots, \hat{x}_i, \cdots,
\hat{x}_j, \cdots, x_{n+1}, y_i])=B(x_j, y_j), i\neq j.
$$
Similarly $B(y_i, y_j)=0, $ for $i\neq j.$ Therefore for any metric
$B_0$ on ${\mathfrak g}_0$, there exists a basis $x_1,$ $ \cdots,$ $
x_{n+1},$ $ y_1,$ $ \cdots,$ $ y_{n+1}$ such that the matrix of
$B_0$ associated with the basis is
$$\left(
      \begin{array}{cccc}
        \lambda I_0 & \mu I_0  \\
     \mu I_0 & 0  \\
        \end{array}
    \right),
    $$
 where $\lambda, \mu\in F, ~\lambda \mu\neq 0, $ $I_0$ is the $(n+1)\times
 (n+1)$- unit matrix. Thanks to Theorem 3.1 in \cite{7}, the metric dimension of ${\mathfrak g}_0$ is $2$.
 Therefore, the metric dimension of ${\mathfrak g}$ is $2m$ and the matrix of any
 metric on ${\mathfrak g}$ is (3.3).
\end{proof}

Concerning the simple case in Theorem 3.7, we have the following
result.

\begin{coro} Let $({\mathfrak g}, {B})$ be a perfect metric $n$-Lie
algebra. Then $\dim {\mathfrak g} \geq m(n+1)$, where $m$ is a
positive integer.
\end{coro}

\section{Classifications of $(n+k)$-dimensional metric $n$-Lie algebras}
In this section we classify $(n+k)$-dimensional metric $n$-Lie
algebras, where $2\leq k \leq n+1$. Theorem 4.3 plays a fundamental
role in the classification, which is based on the results of section
3. First, we give a well-known fact.

\begin{lemma} Let $({\mathfrak g}, B)$ be a metric $n$-Lie algebra. Then there is a
decomposition
$${\mathfrak g}=C_1\oplus {\mathfrak g}_1 ~\mbox{(as the
 orthogonal direct sum of ideals)},\eqno(4.1)$$
where the center of ${\mathfrak g}_1$ is isotropic, and $C_1=0$ or
$C_1$ is a non-degenerate abelian ideal.
\end{lemma}

\begin{theorem} Let $({\mathfrak g}, B)$ be a metric $n$-Lie algebra, where $\dim {\mathfrak g}=n+k$ and $2\leq k\leq n+1.$
Then $\dim {\mathfrak g}^1\leq n+k-1.$ Furthermore if ${\mathfrak
g}$ is unsolvable, then ${\mathfrak g}$ is reductive, that is,
${\mathfrak g}={\mathfrak s}\oplus C({\mathfrak g})$, where
${\mathfrak s}$ is the semisimple ideal of ${\mathfrak g}$.
\end{theorem}

\begin{proof} If ${\mathfrak g}$ is solvable, then ${\mathfrak g}^1=[{\mathfrak g}, \cdots, {\mathfrak g}]$ is a
proper ideal of ${\mathfrak g}$. It implies $\dim {\mathfrak
g}^1\leq n+k-1$.

Assume that ${\mathfrak g}$ is unsolvable. Let ${\mathfrak
g}={\mathfrak s}\oplus {\mathfrak r}$ be the Levi-decomposition,
where ${\mathfrak s}\neq 0$. Since $\dim {\mathfrak s}=n+1$,
${\mathfrak r}$ is a ${\mathfrak s}$-module with $\dim {\mathfrak
r}=k-1 <n+1$. Thanks to \cite{20}, $$[{\mathfrak s}, \cdots,
{\mathfrak s}, {\mathfrak r}]=0.$$
 Therefore,
 $B({\mathfrak s}, {\mathfrak r})=B([{\mathfrak s}, \cdots, {\mathfrak s}], {\mathfrak r})=B({\mathfrak s}, [{\mathfrak s}, \cdots, {\mathfrak s}, {\mathfrak
 r}])=0.$ It implies ${\mathfrak r}\subseteq {\mathfrak s}^{\bot}$. By $\dim {\mathfrak s}+\dim {\mathfrak s}^{\bot}=\dim {\mathfrak g}=\dim {\mathfrak s}+\dim
 {\mathfrak r}$,  ${\mathfrak s}={\mathfrak r}^{\bot}$ is an ideal
 of ${\mathfrak g}$ and $[{\mathfrak r}, {\mathfrak s}, {\mathfrak g}, \cdots, {\mathfrak g}]=0$.

Assume that $[{\mathfrak r}, \cdots, {\mathfrak r}]\neq 0$. Then
$\dim {\mathfrak r}=n$ and there is a basis $e_1, \cdots, e_n $ of
${\mathfrak r}$ such that $[e_1, \cdots, e_n]=e_1.$ Therefore
 $$B(e_1, e_i)=B([e_1, \cdots, e_n], e_i)=B(e_1, (-1)^{i}[e_2, \cdots, e_n, e_i, e_i])=0,
 ~\mbox{ for }~ i\neq 1,$$
  $$B(e_1, e_1)=B([e_1, \cdots, e_n], e_1)=B(e_n, [e_1, \cdots, e_{n-1},
 e_1])=0.$$
That is, $B(e_1, {\mathfrak g})=0.$ This contradicts the
non-degeneracy of $B$. Therefore ${\mathfrak r}=C({\mathfrak g})$.
It follows that ${\mathfrak g}$ is reductive and $\dim {\mathfrak
g}^1=n+1 \leq n+k-1.$
 \end{proof}

\begin{theorem} Let $({\mathfrak g}, B)$ be a non-abelian
$(n+k)$-dimensional metric $n$-Lie algebra with the isotropic
center, where $2\leq k\leq n+1.$ Then $\dim C({\mathfrak g})=k-1$.
\end{theorem}

\begin{proof} Suppose that $\dim C({\mathfrak g})=l$. Thanks to Theorem 4.1 and Lemma
3.1, there exists a basis $\{ e_1, \cdots, e_l,$ $ e_{l+1}, \cdots,
e_{l+t}, e'_1, \cdots, e'_l \}$ of ${\mathfrak g}$ such that the
non-zero products are
$$
B(e_r, e'_s)=\delta_{rs},\quad B(e_{l+i}, e_{l+j})=\delta_{ij}.
\eqno(4.2)
$$
where $1\leq  r, s\leq l, ~ 1\leq i, j\leq t,$ $C({\mathfrak
g})=\langle e_1, \cdots, e_l \rangle$, and ${\mathfrak g}^1=\langle
e_1, \cdots, e_{l+t} \rangle$. By $\dim C({\mathfrak g})+\dim
{\mathfrak g}^1=\dim {\mathfrak g}=n+k,$ we have $l<k.$

If $l=k-2$, then $n+k-l=n+2$, thanks to eqs. (2.3) and (4.2), the
non-zero brackets of basis vectors are
$$C_{i,j}^{1}=[e_{l+1}, \cdots, \widehat{e}_{l+i}, \cdots,
\widehat{e}_{l+j}, \cdots, e_{l+t}, e'_{1}, \cdots, e'_{l}]=\mu_{i,
j}^{1, i}e_{l+i}+\mu_{i, j}^{1, j}e_{l+j}, $$
$$  C_{i,r}^{2}=[e_{l+1}, \cdots, \widehat{e}_{l+i}, \cdots,
e_{l+t}, e'_{1}, \cdots, \widehat{e}'_{r}, \cdots,
e_{l}]=\lambda_{i, r}^{2, r}e_{r}+\mu_{i, r}^{2, i}e_{l+i}, $$
$$ ~ C_{r,s}^{3}=[e_{l+1}, \cdots, e_{l+t}, e'_{1}, \cdots,
\widehat{e}'_{r}, \cdots, \widehat{e}'_{s}, \cdots,
e'_{l}]=\lambda_{r, s}^{3, r}e_{r}+\lambda_{r, s}^{3, s}e_{s}, $$
where $1\leq i, j\leq t, ~ 1\leq r, s\leq l.$ Furthermore for $1\leq
i\neq j\leq t, ~ 1\leq r\neq s\leq l$, we obtain
$$ \mu_{i, j}^{1, i}=(-1)^{t-r-i-1}\lambda_{j, r}^{2, r}, \quad \mu_{i, j}^{1, j}=(-1)^{t-r-j}\lambda_{i, r}^{2, r}, \quad \mu_{i,
r}^{2, i}=(-1)^{-j-i}\mu_{j, r}^{2, j},$$
$$ \lambda_{r, s}^{3,
r}=(-1)^{t-r-i}\mu_{i, s}^{2, i}, \quad \lambda_{r, s}^{3,
s}=(-1)^{t-s-i-1}\mu_{i, r}^{2, i}, \quad\lambda_{i, r}^{2,
r}=(-1)^{-s-r}\lambda_{i, s}^{2, s}. $$
For $1\leq i\neq j\leq t, 1\leq r\neq s\leq l,$
$$
(-1)^{t-r-i}\mu_{i, r}^{2, i}C_{i, j}^{1}=\lambda_{j, r}^{2, r}C_{i,
r}^{2}-\lambda_{i, r}^{2, r}C_{j, r}^{2},\quad
(-1)^{t-r-i}\lambda_{i, r}^{2, r}C_{r, s}^{3}=\mu_{i, s}^{2, i}C_{i,
r}^{2}-\mu_{i, r}^{2, i}C_{i, s}^{2}, $$
$$ (-1)^{s}\mu_{i, s}^{2,
i}\lambda_{j, r}^{2, r}C_{i, r}^{2}+(-1)^{r}\mu_{i, r}^{2,
i}\lambda_{i, s}^{2, s}C_{j, s}^{2}=(-1)^{r}\mu_{i, r}^{2,
i}\lambda_{j, s}^{2, s}C_{i, s}^{2}+(-1)^{s}\mu_{i, s}^{2,
i}\lambda_{i, r}^{2, r}C_{j, r}^{2}. $$ It follows that ${\mathfrak
g}^1$ is spanned by $C_{1, 1}^{2}, C_{1, 2}^{2}, \cdots, C_{1,
l}^{2}, C_{2, l}^{2}, \cdots, C_{t, l}^{2}$. It contradicts $\dim
 {{\mathfrak g}}^{1}=l+t$. Therefore, $l \neq k-2$.

If $l< k-2,$  then $l < n-1$, that is,  $n-l \geq 2$. We obtain an
$(n-l)$-Lie algebra ${\mathfrak g}_0={\mathfrak g}$ (as vector
spaces) with the multiplication given as follows
$$[x_{1}, \cdots, x_{n-l}]_{0}=[x_{1},
\cdots, x_{n-l}, e'_{1}, \cdots, e'_{l}], \quad \forall x_{1},
\cdots, x_{n-l}\in {{\mathfrak g}}_0.\eqno(4.3)$$ It is easy to
check $$B([x_{1}, \cdots, x_{n-l-1}, z]_0, y) = - B([x_{1}, \cdots,
x_{n-l-1}, y]_0, z)$$ for any $x_{1}, \cdots, x_{n-l-1}, y, z\in
{\mathfrak g}_0$. It follows that $({\mathfrak g}_0, B)$ is a metric
$(n-l)$-Lie algebra.

We claim ${\mathfrak g}^1_0=\langle e_{l+1}, \cdots, e_{l+t}
\rangle$ and $C({\mathfrak g}_0)=\langle  e_1, \cdots, e_l, e'_1,
\cdots, e'_{l} \rangle$.

In fact, set $[x_{1}, \cdots,
x_{n-l}]_0=\sum_{r=1}^{l}a_{r}e_{r}+\sum_{j=1}^{t}b_{j}e_{l+j}$,
where $x_{1}, \cdots, x_{n-l}\in {{\mathfrak g}}_0.$
 Thanks to the identity (4.2),  for  $1\leq s\leq l$
$$a_s=B(a_s e_{s}, e'_{s})=B( [x_{1}, \cdots, x_{n-l}]_0 -\sum_{r=1, r\neq s}^{l}a_{r}e_{r}-\sum_{j=1}^{t}b_{j}e_{l+j},
~e'_{s} )=0.
$$
Then ${\mathfrak g}_{0}={C}({\mathfrak g}_{0})\oplus {\mathfrak
g}_{0}^{1}$. Thus $({\mathfrak g}_{0}^{1}, B)$ is a $t$-dimensional
perfect metric $(n-l)$-Lie algebra. By Corollary 3.8, $\dim
{\mathfrak g}_{0}^{1}\geq m(n-l+1)$ for some positive integer $m$.
It is a contradiction.
Therefore, $l=k-1$.
\end{proof}

\begin{theorem} Let $({\mathfrak g}, B)$ be a non-abelian $(n+k)$-dimensional
metric $n$-Lie algebra, where $2\leq k\leq n+1$ and the center is
isotropic. Then there exists a basis $e_{1}, \cdots, e_{k-1},$ $
e_{k}, \cdots, e_{n+1},$ $ e_{n+2}, $ $\cdots,$ $ e_{n+k}$ of
${\mathfrak g}$ such that $C({\mathfrak g})=\langle e_{1}, $
$\cdots, $ $e_{k-1} \rangle, $ ${\mathfrak g}^1=\langle e_{1},
\cdots,$ $ e_{k-1}, e_{k}, $ $\cdots, $ $e_{n+1}\rangle $, and
$$B(e_{r}, e_{n+1+s})=\delta_{rs},\quad B(e_{i}, e_{j})=\delta_{ij}, \eqno(4.4)
$$ where $~1\leq r, s\leq k-1; ~k\leq i, j\leq n+1,$   the multiplication table
under the basis is
$$(1) ~\begin{array}{l}
 \left\{\begin{array}{l}
{[}e_{k}, \cdots, \widehat{e}_{i}, \cdots, e_{n+1}, e_{n+2}, \cdots,
e_{n+k}]=(-1)^{n+i}ae_{i}, \quad k\leq i\leq n+1,  \\
 {[}e_{k}, \cdots, e_{n+1}, e_{n+2}, \cdots, \widehat{e}_{n+1+r},
\cdots, e_{n+k}]=(-1)^{r+1}ae_{r}, \quad 1\leq r\leq k-1,
\end{array}\right.
\end{array}
$$
where $a\in {\mathbb C}$ and $a\neq 0.$
\end{theorem}

\begin{proof} By the proof of Theorem 4.3, there exists a basis $
e_{1}, \cdots, e_{k-1},$ $ e_{k}, \cdots, e_{n+1},$ $ e_{n+2}, $
$\cdots,$ $ e_{n+k}$ of ${\mathfrak g}$ satisfying the identity
(4.4), where $C({\mathfrak g})=\langle e_{1}, $ $\cdots, $
$e_{k-1}\rangle, $ ${\mathfrak g}^{1}=\langle e_{1}, \cdots,$ $
e_{k-1}, e_{k}, $ $\cdots, $ $e_{n+1} \rangle$. Let
$$[e_{k}, \cdots, \widehat{e}_{q}, \cdots,
e_{n+k}]=\sum_{r=1}^{k-1}a_{qr}e_r+\sum_{i=k}^{n+1}b_{qi}e_{i},
\quad k\leq q\leq n+k.$$ The theorem follows from the direct
computation according to eqs. (2.3) and (4.4).
\end{proof}

\begin{theorem} Let $({\mathfrak g}, B)$ be a non-abelian $(n+k)$-dimensional
metric $n$-Lie algebra, where $2\leq k\leq n+1$. Then $({\mathfrak
g}, B)$ is one of the following cases up to isomorphisms.
\begin{enumerate}
\item If the center of ${\mathfrak g}$ is isotropic, then ${\mathfrak g}$ is
the case (1) in Theorem 4.4.
\item
Assume that the center of ${\mathfrak g}$ is non-isotropic. If
${\mathfrak g}$ is reductive, there exists a basis $x_{1}, \cdots, $
$x_{k-1},$ $ e_{1}, \cdots,$ $ \cdots,~e_{n+1}$ of ${\mathfrak g}$,
satisfying
$$\begin{array}{l}
 \left\{\begin{array}{l}
 B(x_p, x_q)=\delta_{pq}, ~1\leq p, q\leq k-1;\\
 B(x_p, e_i)=0,~ 1\leq p\leq k-1, ~1\leq i\leq n+1;\\
  B(e_i,
e_j)=\delta_{ij},  1\leq i, j\leq n+1,
\end{array}\right.
\end{array}
\eqno(4.5)
$$  and the multiplication table is given as follows:
$$ (2)~ [e_{1}, \cdots,
\widehat{e}_{r}, \cdots, e_{n+1}]=(-1)^{r+1} c e_{r},$$ where $c\in
{\mathbb C}, c\neq 0, 1\leq r\leq n+1$.

If ${\mathfrak g}$ is non-reductive, that is, $C({\mathfrak g})\cap
{\mathfrak g}^1\neq 0$. Then there exists a basis $x_{1}, \cdots,
x_{l},$ $ e_{1}, \cdots,$ $ e_{k_1},$ $ \cdots,~$ $e_{n+1}$ $,
\cdots, $ $ e_{n+k-l }$ of ${\mathfrak g}$ satisfying $C({\mathfrak
g})=\langle x_{1}, \cdots, x_{l}, e_{1}, \cdots, e_{k_1-1} \rangle$,
${\mathfrak g}^1=\langle e_{1}, \cdots, e_{k_1-1}, e_{k_1}, \cdots,
e_{n+1} \rangle,$
$$
\begin{array}{l}
 \left\{\begin{array}{l}
B(x_{p}, x_{q})=\delta_{pq}, ~1\leq p, q\leq l;\\
B(x_p, e_m)=0,~  ~1\leq p\leq l,~ 1\leq m\leq n+k-l;\\
 B(e_{r}, e_{s})=\delta_{rs},~~k_1\leq r, s\leq n+1;\\
B(e_{i}, e_{n+1+j})=\delta_{ij},~~ 1\leq i,
j\leq k_1-1, \\
B(e_{n+1+i}, e_{n+1+j})=0,~~ 1\leq i, j\leq k_1-1,
\end{array}\right.
\end{array}
\eqno(4.6)$$ and the multiplication table is as follows
$$ (3)~\begin{array}{l}
 \left\{\begin{array}{l}
{[}e_{k_1}, \cdots, \widehat{e}_{i}, \cdots, e_{n+1}, e_{n+2},
\cdots,
e_{n+k-l}]=(-1)^{n+i}~a~ e_{i}, ~~~~ k_1\leq i\leq n+1;  \\
 {[}e_{k_1}, \cdots, e_{n+1}, e_{n+2}, \cdots, \widehat{e}_{n+1+r},
\cdots, e_{n+k-l}]=(-1)^{r+1}~a~ e_{r}, ~~~~1\leq r\leq k_1-1,
\end{array}\right.
\end{array}
$$
where $1\leq l < k-1, k_1=\dim (C({\mathfrak g})\cap {\mathfrak
g}^1)=k-l-1\geq 1$, $a\in {\mathbb C}, a\neq 0.$
\end{enumerate}
\end{theorem}
\begin{proof}  The result (1) follows the Theorem 4.4. The case
(2) and identity (4.5) follow from the direct computation by
\cite{1,19} and the identity (2.3). Lastly assume that the center of
${\mathfrak g}$ is non-isotropic and ${\mathfrak g}$ is
non-reductive. The theorem follows from Lemma 4.1 and Theorem
4.3.\end{proof}

Summarizing above result, we obtain the following result.

\begin{coro} Let ${\mathfrak g}$ be a
$(n+k)$-dimensional metric $n$-Lie algebra with $2\leq k\leq n+1$.
Then $\dim {\mathfrak g}^1=0$ or $\dim {\mathfrak g}^1=n+1.$
\end{coro}

\section*{Acknowledgments}

The first and second authors are partially supported by NSF of China
(10871192) and NSF of Hebei Province, China (A2010000194).

\bibliography{}

\end{document}